\documentclass[a4paper, 12pt]{article}

\usepackage[papersize={210mm,297mm},top=30mm, bottom=30mm, left=26mm , right=26mm]{geometry}

\usepackage{amsmath,amssymb}
\usepackage{amsthm}
\usepackage{todonotes}
\usepackage{float}

\usepackage{listings}
\usepackage{xcolor} 

\usepackage[colorlinks,citecolor=magenta,urlcolor=blue]{hyperref}

\usepackage{amsthm}

\usepackage{dsfont}

\newtheorem{theorem}{Theorem}
\newtheorem{proposition}{Proposition}
\newtheorem{remark}{Remark}

\DeclareMathOperator*{\var}{\rm{Var}} 

\newcommand{\Prob}{\mathds{P}}
\newcommand{\E}{\mathds{E}}
\newcommand{\EI}{\mathrm{Ei}}

\usepackage{authblk}

\begin{document}
%%%%%%%%%%%% TITLE PAGE %%%%%%%%%%%%%%%%%%%

\title{Exponential Integral Solutions for Fixation Time\\ in Wright-Fisher Model With Selection}
\author[1]{Vincent Runge\thanks{Email: vincent.runge@univ-evry.fr}}
\author[1,2]{Arnaud Liehrmann}
\author[1]{Pauline Spinga}
\affil[1]{\small Universit\'e Paris-Saclay, CNRS, Univ Evry, Laboratoire de Math\'ematiques et Mod\'elisation d'Evry, 91037, Evry-Courcouronnes, France.}
\affil[2]{\small Universit\'e Paris-Saclay, CNRS, INRAE, Universit\'e  Evry, Institute of Plant Sciences Paris-Saclay (IPS2), 91190, Gif sur Yvette, France.}

%\Plainauthor{Vincent Runge}

% List of institutions
%\affil[1]{LaMME - Laboratoire de Math\'ematiques et Mod\'elisation d'Evry.\newline UEVE - Universit\'e d'Evry-Val-d'Essonne.}
\date{}
\maketitle

% Abstract of the paper
\begin{abstract}
In this work we derive new analytic expressions for fixation time in Wright-Fisher model with selection. The three standard cases for fixation are considered: fixation to zero, to one or both. Second order differential equations for fixation time are obtained by a simplified approach using only the law of total probability and Taylor expansions. The obtained solutions are given by a combination of exponential integral functions with elementary functions. We then state approximate formulas involving only elementary functions valid for small selection effects. The quality of our results are explored throughout an extensive simulation study. We show that our results approximate the discrete problem very accurately even for small population size (a few hundreds) and large selection coefficients.
\end{abstract}


%%%%%%%%%%%%%%%%%%%%%%%%%%%%%%%%%%%%%%%%%%%%%%%%%%%%%%%%%%
\section{Introduction}

%\paragraph{What is Wright-Fisher?} 
The Wright-Fisher (WF) model is one of the simplest model of genetic drift in population genetics. It has been developed independently by Sewall Wright \cite{wright1931evolution} and R.A Fisher \cite{sir1999genetical} in the 1930s. This model describes the dynamic of allele transmission over time for haploid or diploid individuals based on strict assumptions: non-overlapping generations, constant population size over time, random (panmictic) mating.

%\paragraph{There are many applications to this day} 
One of the first applications to the WF model was proposed by Buri et al$.$ \cite{buri1956gene} who studied the distribution of gene frequency under selection due to a mutation in Drosophilia Melanogaster populations. This idea was undertaken later by Grifford et al$.$ \cite{gifford2013model} who studied the probability of fixation in mutated strains of fungus Aspergillus nidulans impacting their viability. More recent examples among many others include the analysis of genomics data \cite{terhorst2015, fares2015experimental, tataru2015inference, ferrer2016, taus2017, jonas2016}, reflecting the relevance of such a simple model to this day. 

%\paragraph{More complex models}
In addition to genetic drift, the WF model can account for the effects of other evolutionary forces such as mutation \cite{Gutenkunst2009, Luki2012,Steinrcken2014}, selection \cite{Mathieson2013, Gompert2015,Bollback2008,Malaspinas2012,Steinrcken2014,Vitalis2014,Lacerda2014}, the effect of demographic forces including the variation in population size through time and migration connecting populations \cite{Mathieson2013, Gutenkunst2009, Luki2012,Vitalis2014}.

%\paragraph{We study fixation time - What is known and what we found} 
In population genetics, one focuses on different quantities of interest such as heterozygosity or fixation probabilities \cite{shafiey2017exact}. Here we focus on the average time until fixation (or simply fixation time) in WF model. Many solutions for fixation time have been found using the diffusion approximation method \cite{kimura1962probability}. The expression of the fixation time for WF model with no additional effect is well-known \cite{etheridgediffusion, tran2013introduction}. Solutions for fixation of a rare mutant gene (excluding the case of potential loss) by random drift and with selection \cite{kimura1969average} or with an additional mutation rate \cite{kimura1980average} have also been derived. Moreover, the variance for fixation with no loss has been found in \cite{narain1970note} and \cite{maruyama1972average} with a selection coefficient.

In this work we intended to contribute to the study of the fixation time for WF model with selection by stating new analytic expressions. We analyzed the three possible fixation times: fixation to zero, to one or both. To the best of our knowledge, only some partial results have been described in literature: fixation to one for rare mutant for all initial probabilities in complex integral form \cite{kimura1969average} and for null initial probability (only) \cite{koopmann2017fisher} in exponential integral form. No author has yet highlighted the use of the exponential integral function for describing its complete solution. 

%\paragraph{Outlines} 
This paper is organized as follows. In Section~\ref{sec:2} we present the WF model with selection effect and the derivation of its associated ordinary differential equations. In Section~\ref{sec:3} we find expressions for fixation time using exponential integral (EI) functions. In Section~\ref{sec:4} we study approximation results to bypass the use of EI functions for small selection effect. We eventually provide an extensive simulation study in Section~\ref{sec:5} to explore the validity range of our theoretical results with respect to population size and selection effect strength. 

\section{The Wright-Fisher Model With Selection}
\label{sec:2}

\subsection{Context and Notations}

%\paragraph{Haploid - diploid ($h = 1/2$)}
We consider an initial population size of $N$ haploid individuals and study a single locus with two possible alleles: $A$ and $a$. This model is equivalent to the study of diploid individuals who can either be homozygote ($AA$ - $aa$) or heterozygote ($Aa$ - $aA$), considering that $A$ and $a$ are neither dominant nor recessive. In that case, the dominance parameter of heterozygous individuals is $h = 1/2$ \cite[chapter 5.3]{etheridge2011some}. 

%\paragraph{Effective size}
Notice that, in the literature \cite{braude2009understanding, wang2016prediction}, while studying genetic drift, an effective population size is usually determined. Denoted $N_e$, it is defined as the size of an ideal population that would have the same rate of loss of genetic diversity - or that would lose heterozygosity at an equal rate - than the real studied population. The determination of $N_e$ is not the subject of this article, we consider that our population size is equal to the effective population size ($N~=~N_e$).

We introduce a selection effect $s$ that affects favorably viability of carriers of allele $A$: their chance of survival is increased by a factor $1 + s$. There are many possibilities to introduce selection effects \cite{shafiey2017exact}. We chose to keep $s$ constant for all generations and we also define coefficient $\alpha = 2sN$ normalized by population size as usually done in the literature \cite{kimura1969average, tataru2017statistical, bollback2008estimation, etheridge2011some} (estimation of $s$ from data in \cite{bollback2008estimation}).

%\paragraph{Some definitions}

Let $x_t$ be the number of alleles A in generation $t$, $x_t \in \{0,\dots, N\}$, and $y_t = x_t/N$ its proportion\footnote{We chose notation $y_t$ instead of a more natural $p_t$ to avoid confusion with initial allele A frequency, further denoted $p$.}: they are realizations of random variables $X_t$ and $Y_t$ respectively. The frequency $y_{t+1}$ of alleles A obtained at time $t+1$ depends on frequency $y_{t}$ at time $t$, but is moreover skewed by the selection effect, that is, we define a fictional proportion:
 $$\pi_{t} = \frac{x_t+s x_t}{ N + s x_t} = \frac{(1+s)y_t}{1+ sy_t}\,,$$
 where the population of allele A is increased by an imaginary quantity $s x_t$. The distribution of alleles in the new generation is then generated by the binomial distribution $X_{t+1} \backsim \mathcal{B}(N,\pi_{t})$. The generation $t+1$ represents the offspring of the generation $t$ obtained by random sampling with replacement. We also define the variation of the proportion $Y_t$ over time by $\Delta P_t = Y_{t+1} - Y_{t}$. The transition probabilities are defined as 
$$ p_{ij} = \Prob(X_{t+1} = j\mid X_t = i) =\binom{N}{j} \Big( \frac{(1+s)i}{N+si}\Big)^{j} \Big(1- \frac{(1+s)i}{N+si}\Big)^{N-j}\,,$$
for all integers $t$, integers $i$ in $\left \{ 1,..., N-1\right \}$ and $j$ in $\left \{ 0,..., N\right \}$.

\begin{remark}
$\left( X_{t} \right)_{t \ge 0}$ is an absorbing Markov chain~\cite{kemeny1976laurie} where $p_{ij}$ is the probability of transition from state $i$ to $j$ and $X_t = 0$ and $X_t = N$ are the absorbing states.
\end{remark}

\subsection{Fixation Time}

The fixation time is defined as the time $T=T_p$ needed by the Markov chain $(X_t)_{t \ge 0}$ to reach $0$ or $N$ with initial population $X_0 = pN$ in $\{0,\dots, N\}$, that is:
$$T_p = \min_{t \ge 0}\{t\,\,\hbox{such that}\,\,  X_t = 0 \,\, \hbox{or} \,\,X_t = N\,\,\hbox{with}\,\, X_0 =   pN\}\,.$$ 
This is a random variable returning the number of iterations needed to get an homogeneous population. It is parametrized by its initial allele A frequency $p$. Our goal is to search for analytic expressions for its expectation when we consider the limit case for infinite population size $N$, leading to a continuous variable $p$ in $[0,1]$. We then re-define $T_p$ through the frequencies $Y_t$ to remove size-dependency: 
$$T_p = \min_{t \ge 0}\{t\,\,\hbox{such that}\,\,  Y_t = 0 \,\, \hbox{or} \,\,Y_t = 1\,\,\hbox{with}\,\, Y_0 =   p\}\,.$$
We decompose its expected value by the law of total expectation and introduce new notations as follows:
$$
\begin{aligned}
  \E[T_p] &= \Prob (Y_{T_p} = 0) \E[T_p \mid Y_{T_p} = 0]+ \Prob (Y_{T_p} = 1) \E[T_p \mid Y_{T_p} = 1]\,,\\
  m(p) &= (1-u(p)) m_0(p)+ u(p)m_1(p)\,,\\
\end{aligned}
$$
with expectation function $p \mapsto m(p) = \E[T_p]$, fixation probability function $p \mapsto u(p) = \Prob (Y_{T_p} = 1)$, expected fixation time to zero and one defined as $p \mapsto m_1(p) = \E[T_p \mid Y_{T_p} = 1]$ and $p \mapsto m_0(p) = \E[T_p \mid Y_{T_p} = 0]$, respectively. Function $m_1(\cdot)$ is often called the fixation time of a mutant gene excluding the case of potential loss. We also define $M_1(p) = u(p)m_1(p)$ and $M_0(p) = (1-u(p))m_0(p)$.
These quantities can be described in terms of infinite sums:
$$
 \left\{
    \begin{aligned}
  u(p) &= \Prob (Y_{T_p} = 1)=\sum_{t=0}^{\infty} \Prob(T_p = t \,,\, Y_{T_p} = 1)\,,\\
  M_0(p) &=\sum_{t=0}^{\infty} t \Prob(T_p=t \,,\, Y_{T_p} = 0)\,,\\
  M_1(p) &=\sum_{t=0}^{\infty} t \Prob(T_p=t \,,\, Y_{T_p} = 1)\,,\\
  m(p) &= \E[T_p]=\sum_{t=0}^{\infty} t \Prob(T_p=t) = M_0(p) + M_1(p)\,,\\
\end{aligned}
\right.
$$
and we can now derive equations binding neighboring values of each function to itself using the transition probability operator $\Prob_{p \rightarrow p + \Delta p} = \Prob(Y_{1}=p + \Delta p  \mid Y_{0}=p)$, which is the density probability of the transition in one time step from allele probability $p$ to allele probability $p+\Delta p$.

\begin{proposition}
\label{masterEquations}
Functions $u(\cdot)$, $M_0(\cdot)$, $M_1(\cdot)$ and $m(\cdot)$ satisfy the following equations for initial probability $p$ (sometimes called master equations) :
$$
 \left\{
    \begin{aligned}
   u(p) &=  \int_{\Delta p}\Prob_{p \rightarrow p + \Delta p} u(p + \Delta p)d(\Delta p) \,,\\
    M_0(p) &= 1-u(p) + \int_{\Delta p} \Prob_{p \rightarrow p + \Delta p}M_0(p + \Delta p)d(\Delta p) \,,\\
      M_1(p) &= u(p) + \int_{\Delta p} \Prob_{p \rightarrow p + \Delta p}M_1(p + \Delta p)d(\Delta p) \,,\\
      m(p) &= 1 + \int_{\Delta p} \Prob_{p \rightarrow p + \Delta p}m(p + \Delta p)d(\Delta p) \,,\\
\end{aligned}
\right.
$$
with integration performed with respect to the Lebesgue measure.
\end{proposition}

\begin{proof}
We first provide details for the last equation. The proof is based on the law of total probability:
$$\Prob(T_p=t)  = \int_{\Delta p} \Prob_{p \rightarrow p + \Delta p}\Prob(T_p=t \mid Y_{1}=p + \Delta p)d(\Delta p)\,.$$
We multiply this equality by $t$ and sum for all integers to get $m(p)$ in left hand side:
$$
    \begin{aligned}
      m(p) &= \int_{\Delta p} \Prob_{p \rightarrow p + \Delta p}\sum_{t=0}^{\infty} t\Prob(T_p=t\mid Y_{1}=p + \Delta p)d(\Delta p)\,,\\
       &=\int_{\Delta p} \Prob_{p \rightarrow p + \Delta p}\sum_{t=1}^{\infty} (t-1 +1)\Prob(T_p=t\mid Y_{1}=p + \Delta p)d(\Delta p)\,,\\
    &=\int_{\Delta p} \Prob_{p \rightarrow p + \Delta p}\sum_{t=1}^{\infty} (t-1)\Prob(T_p=t\mid Y_{1}=p + \Delta p)d(\Delta p) + 1\,,\\
     &=\int_{\Delta p} \Prob_{p \rightarrow p + \Delta p}\sum_{t=1}^{\infty} (t-1)\Prob(T_{p +\Delta p}=t-1)d(\Delta p) + 1\,,\\
     &=\int_{\Delta p} \Prob_{p \rightarrow p + \Delta p}\sum_{t=0}^{\infty} t\Prob(T_{p +\Delta p}=t)d(\Delta p) + 1\,,\\
            &=  \int_{\Delta p} \Prob_{p \rightarrow p + \Delta p} m(p + \Delta p)d(\Delta p)  + 1\,.\\
    \end{aligned}
$$
Results for $M_0(\cdot)$ and $M_1(\cdot)$ are obtained by considering the same proof replacing $\Prob(T_p=t)$ by $\Prob(T_p=t, Y_{T_p} = 0)$ or $\Prob(T_p=t, Y_{T_p} = 1)$ in the law of total probability. For equation in $u(\cdot)$ (the first equation) with sum all probabilities $\Prob(T_p=t, Y_{T_p} = 1)$ without the multiplicative factor $t$.
\end{proof}

\begin{remark}
With discrete proportions for the initial probability and a finite size population ($p = i/N$, $i=0,\dots,N$), equations in Proposition \ref{masterEquations} are valid with integration with respect to counting measures. They can be solved exactly by matrix inversion.
\end{remark}

\begin{remark}
Similar results are available for the second moment: $w(p) = \E[T^2_p]=\sum_{t=0}^{\infty} t^2 \Prob(T_p=t)$, leading to equation
$$w(p) = 1 + 2m(p) + \int_{\Delta p}\Prob_{p \rightarrow p + \Delta p} w(p + \Delta p)d(\Delta p) \,,$$
and we would get the variance $v(p)$ as $w(p) - (m(p))^2$. We did not succeed in finding explicit solutions for second moments. Explicit solutions are known, to the best of our knowledge, only in non-selection mode \cite{narain1970note, maruyama1972average}.
\end{remark}

%The fixation time can be determined discretely, using matrix. In this way, fixation times are calculated for every initial quantities of allele $A$. However, the continuous method considers $Y_t$ as a continuous value over time. %Regarding the literature \cite{Ref2}, 

\subsection{Second Order Differential Equations}

Master equations can be approximated by Taylor expansions around small variations in proportions $p$.

\begin{proposition}
\label{equadiffm1}
Equations in Proposition \ref{masterEquations} can be approximated by the following second order differential equations:
\begin{equation}
\label{secondorderequation}
 \left\{
    \begin{aligned}
 &  u''(p) + \alpha \, u'(p)  = 0 \,,\\
 &  M_0''(p) + \alpha \, M_0'(p) + \frac{2N}{p(1-p)}(1-u(p)) = 0 \,,\\
 &  M_1''(p) + \alpha \, M_1'(p) + \frac{2N}{p(1-p)}u(p) = 0 \,,\\
 &  m''(p) + \alpha \, m'(p) + \frac{2N}{p(1-p)} = 0\,, \\
\end{aligned}
\right.
\end{equation}
with natural boundary conditions: $u(0) =0$, $u(1) = 1$, $M_0(0) = 0$, $M_1(1) = 0$ and $m(0) = m(1) = 0$, which leads to $M_0(1) =(1-u(1))m_0(1)= 0$ and $M_1(0) =u(0)m_1(0)= 0$.
\end{proposition}

\begin{proof}
A detailed proof is given in \ref{app:diffequation}. All the results are based on the following expansion given here for a generic function $f$:
 $$\begin{aligned}
 \int  \Prob_{p \rightarrow p+q} &f(p + \Delta p)d(\Delta p)\\
 &=  \int \Prob_{p \rightarrow p + \Delta p} \left (f(p) + \frac{f'(p)}{1!}\Delta p + \frac{f''(p)}{2!}\Delta p^2 + O(\Delta p^3)  \right )d(\Delta p)\\
 &\approx  f(p) + \E[\Delta P]f'(p) + \frac{1}{2}\E[(\Delta P)^2]f''(p)\,.
  \end{aligned}$$
\end{proof}

Solution for unknown function $u(\cdot)$ in first equation is easy to found (see \cite{kimura1962probability}): it is given by formula $u(p) = \frac{1-e^{-\alpha p}}{1 - e^{-\alpha}}$ and we also have $1-u(p) = \frac{1-e^{\alpha (1-p)}}{1 - e^{\alpha}}$.

\section{Analytic Solutions for Fixation Time}
\label{sec:3}

\subsection{The Exponential Integral Function}

Our solutions are based on the exponential integral function  $x \mapsto \EI(x)$ from $\mathbb{R}^*$ to $\mathbb{R}$ which is a special function defined by an integral (understood in terms of the Cauchy principal value):
$$\EI(x) = -\int_{-x}^{+ \infty} \frac{e^{-t}}{t} dt = \int_{-\infty}^{x} \frac{e^{t}}{t} dt\,.$$
We have the following relation:
\begin{equation}
\label{EIintegral}
 \EI(x) = \ln |x| e^{x} - \int_{0}^{x} \ln |z| e^{z} dz + \gamma\,,  
\end{equation}
with $\gamma$ the Euler-Mascheroni constant ($\gamma \approx 0.5772156649$) and integral
\begin{equation}
\label{EIintegral2}
\int e^{-\alpha x} \EI(\alpha x) dx = \frac{\ln \alpha x}{\alpha} + \frac{e^{-\alpha x} \EI(\alpha x)}{\alpha} \,.
\end{equation}
The shape of the EI function is presented in Figure~\ref{fig:EI}.
\begin{figure}[!h]
\center\includegraphics[width=0.6\textwidth]{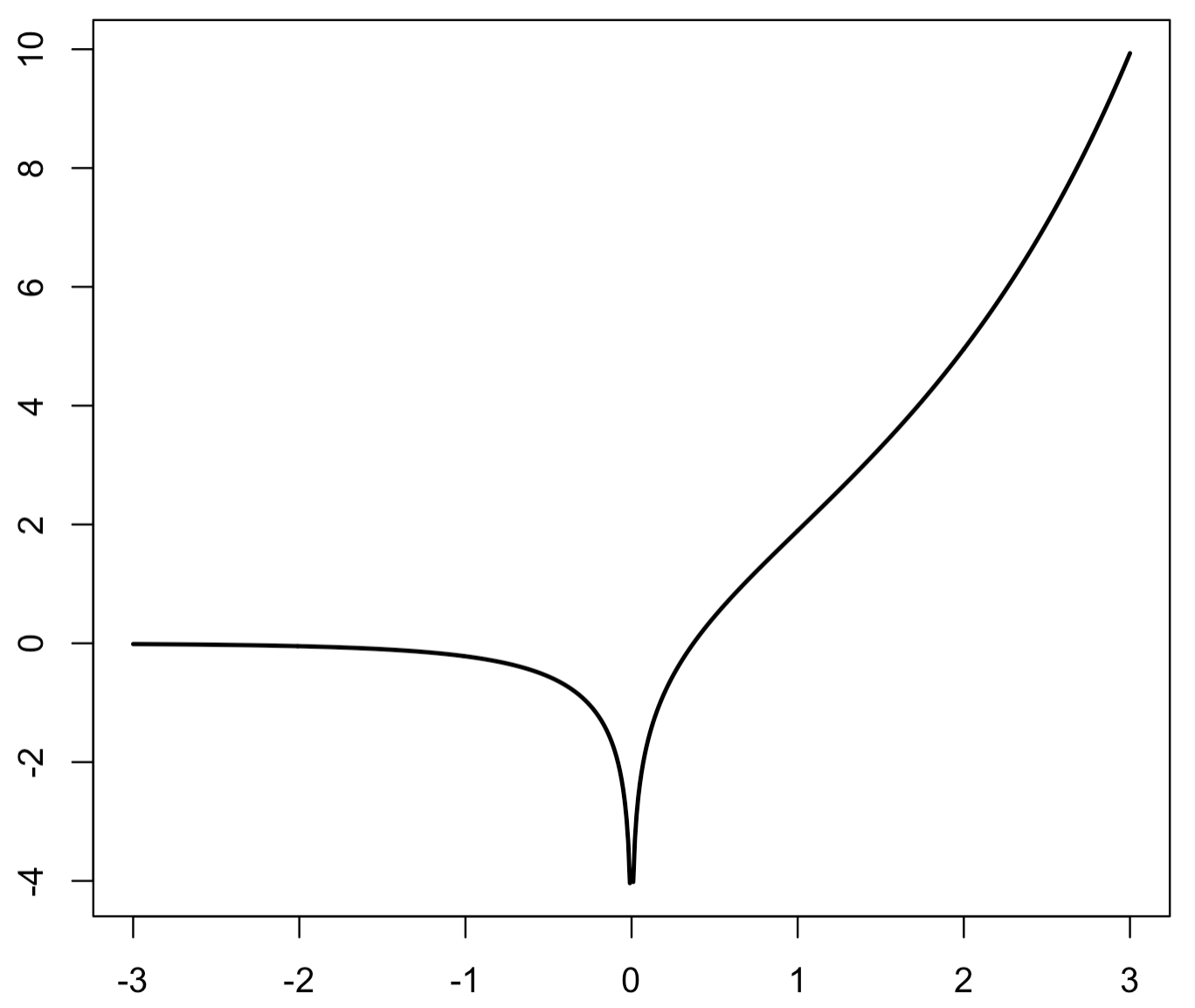} 
\caption{Shape of the exponential integral function near zero.}
\label{fig:EI}
\end{figure}

\subsection{Time to Fixation}

We start by solving the differential equation in Proposition \ref{equadiffm1} for function $m(\cdot)$ and determine an analytic expression.

\begin{theorem}
\label{theorem1}
Time to fixation to both zero and one with initial allele frequency $p$ and a selection effect of order $\alpha = 2sN$ is given by formula
%$$m(p)  = \frac{N}{\alpha}\Big( e^{-2 \alpha p} \EI(2 \alpha p) - e^{-2 \alpha (p-1)}\EI(2 \alpha (p-1)) + \ln(1-p) - \ln(p)$$
%$$+\frac{1-e^{-2 \alpha p}}{1 - e^{-2 \alpha}}C(\alpha) -\frac{1-e^{-2 \alpha (p-1)}}{1 - e^{2 \alpha}} C(-\alpha)\Big)\,,$$
%with $C(\alpha) =  \gamma +  \ln(|2\alpha|) - e^{-2 \alpha} \EI(2 \alpha)$
%that can be also written as
$$m(p)  = 2N\Big( \omega(\alpha,p) + \omega(-\alpha,1-p)\Big)\,,$$
with 
\begin{align*}
\alpha \, \omega(\alpha,p) & =  u(p)\Big( \ln |\alpha| + \gamma - e^{- \alpha} \EI(\alpha)  \Big) - \Big( \ln |\alpha p| + \gamma - e^{-\alpha p} \EI(\alpha p)  \Big)\\
& =  \Big[\frac{1-e^{- \alpha p}}{1 - e^{-\alpha x}}\Big( -e^{-\alpha x} \EI(\alpha x) +  \ln |\alpha x| +\gamma \Big)\Big]_{x=p}^{x=1}\\
& =  \Big[\frac{1-e^{- \alpha p}}{1 - e^{- \alpha x}}\Big( e^{- \alpha x}\int_{0}^{ \alpha x} \ln |z| e^{z}dz  \Big)\Big]_{x=p}^{x=1}\,,
\end{align*}
\end{theorem}

\begin{proof}
Equation in $m(\cdot)$ in Proposition \ref{masterEquations} can be written as
$$ 
     \left [e^{\alpha p} m'(p)  \right ]'= -2N\frac{e^{\alpha p}}{p(1-p)} = -2N\Big(\frac{e^{\alpha p}}{p}+e^{\alpha}\frac{e^{-\alpha (1-p)}}{1-p}\Big)\,,$$
which can be integrated using the $\EI$ function:
   $$e^{\alpha p} m'(p)  = 2N\Big(-\EI({\alpha}p) + e^{{\alpha}}\EI({\alpha}(p-1))\Big) + C_1\,.
$$
We obtain by a second integration using \ref{EIintegral2} (see also \cite{geller1969table}):
$$m(p)  =  \frac{2N}{\alpha}\left [e^{- \alpha p} \EI(\alpha p) - e^{- \alpha (p-1)}\EI(\alpha (p-1)) + \ln(1-p) - \ln(p)  \right ] $$
$$+ \, C_2 - C_1\frac{e^{-\alpha p}}{\alpha}\,. $$

We determine the constants using boundary values $m(0)=m(1)=0$ and the limit $\lim\limits_{x \rightarrow 0} \Big(\EI(\alpha x) - \ln|x| \Big) = \gamma + \ln  |\alpha |$ and we get:
$$
\left\{
    \begin{aligned}
    0 &=C_2 - \frac{C_1}{ \alpha} + \frac{2N}{\alpha}\Big(\gamma + \ln|\alpha | - e^{ \alpha}\EI(-\alpha)\Big) \,,\\
    0 &= C_2 - C_1\frac{e^{- \alpha}}{ \alpha} - \frac{2N}{\alpha}\Big(\gamma + \ln|\alpha | - e^{-\alpha}\EI( \alpha)\Big)\,.\\
        \end{aligned}
        \right.
$$
This system is easy to solve for constants $C_1$ and $C_2$ to find the formula presented in the theorem. The integral expression for $\omega(\alpha,p)$ comes from relation (\ref{EIintegral}) as the terms involving $\gamma$ cancel out.
\end{proof}

\begin{remark}
Both functions $p \mapsto 2N\omega(\alpha,p)$ and $p \mapsto 2N\omega(-\alpha,1-p)$ are solution to its associated equation in (\ref{secondorderequation}) as this equation is invariant by transformation $(\alpha,p) \mapsto (-\alpha, 1-p)$.
\end{remark}

\subsection{Time to Zero and Time to One}

Replacing $M_0(\cdot)$ by $(1-u(\cdot))m_0(\cdot)$ and $M_1(\cdot)$ by $u(\cdot) m_1(\cdot)$ in equations (\ref{secondorderequation}) we get the following equations:
\begin{equation}
\label{secondorderequation2}
 \left\{
    \begin{aligned}
 &  m_0''(p) + \alpha \frac{1 + e^{\alpha (1-p)}}{1 - e^{\alpha (1-p)}}\, m_0'(p) + \frac{2N}{p(1-p)} = 0 \,,\\
 &  m_1''(p) + \alpha \frac{1 + e^{-\alpha p}}{1 - e^{-\alpha p}}\,\, m_1'(p) + \frac{2N}{p(1-p)} = 0 \,,\\
\end{aligned}
\right.
\end{equation}
showing the simple relation $m_0(p) = m_1(1-p)$. Therefore, we only need to solve the equation for $M_1(\cdot)$, as $M_0(p) = (1-u(p))m_1(1-p) = \frac{1-u(p)}{u(1-p)}M_1(1-p)$ leading to:
\begin{equation}
\label{M0M1}
 M_0(p) = \frac{1-e^{\alpha (1-p)}}{1-e^{-\alpha (1-p)}} \frac{1 - e^{-\alpha}}{1 - e^{\alpha}} M_1(1-p) = e^{-\alpha p} M_1(1-p)\,.
\end{equation}

We now derive the complete solution for function $m_1(\cdot)$ by solving its differential equation in $M_1(\cdot)$ in Proposition \ref{equadiffm1}. This result is already presented in \cite{kimura1969average} in integral form.

\begin{theorem}
\label{theorem2}
Time to fixation excluding allele $A$ loss with initial allele frequency $p$ and a selection effect of order $\alpha = 2sN$ is given by formula:
$$m_1(p)  = \frac{2N}{\alpha }\Big(e_1(p) + e_2(p) + e_3(p) + \frac{C}{u(p)} + \gamma + \ln |\alpha |  - e^{-\alpha}\EI(\alpha)  )\Big)\,,$$
with 
$$
\left\{
\begin{aligned}
& e_1(p) = \frac{e^{-\alpha p} \EI(\alpha p) - e^{\alpha (1-p)} \EI(\alpha (p-1))}{1-e^{-\alpha p}}\,, \\
& e_2(p) = \frac{\EI(-\alpha p) - e^{- \alpha}\EI(\alpha (1-p)) }{1-e^{-\alpha p}} \,, \\
& e_3(p) = \frac{1 + e^{-\alpha p}}{1-e^{-\alpha p}}\ln \Big(\frac{1-p}{p}\Big)\,, \\
& C =  \frac{e^{-\alpha}}{1-e^{-\alpha}} \Big( 2(\gamma + \ln |\alpha |) -e^{\alpha}\EI(-\alpha) - e^{-\alpha}\EI(\alpha)\Big) \,.
\end{aligned}
\right.
$$
\end{theorem}

The proof is straightforward: it is similar to the proof of Theorem~\ref{theorem1} by solving the equation for $M_1(\cdot)$ in Proposition~\ref{equadiffm1} and then dividing by function $u$ to get the result. We can use relation (\ref{M0M1}) to notice that it is quite easy to come back to solution of Theorem~\ref{theorem1} using relations
$$
\left\{
\begin{aligned}
& u(p)e_1(p) + e^{- \alpha p} u(1-p)e_1(1-p) = 0 \,, \\
& u(p)e_2(p) + e^{- \alpha p} u(1-p)e_2(1-p) = e^{-\alpha p} \EI(\alpha p) - e^{\alpha (1-p)} \EI(\alpha (p-1)) \,, \\
& u(p)e_3(p) + e^{- \alpha p} u(1-p)e_3(1-p) = \ln \Big(\frac{1-p}{p} \Big) \,.\\
\end{aligned}
\right.
$$

We plot the shape of functions $m(\cdot)$ and $m_1(\cdot)$ for various values of the selection coefficient in Figure~\ref{fig:timeToFixation}. Population size $N$ is only present as a multiplicative factor for all fixation time formulas. It does not affect the general shape of these functions.

\begin{figure}[!h]
\center\includegraphics[width=\textwidth]{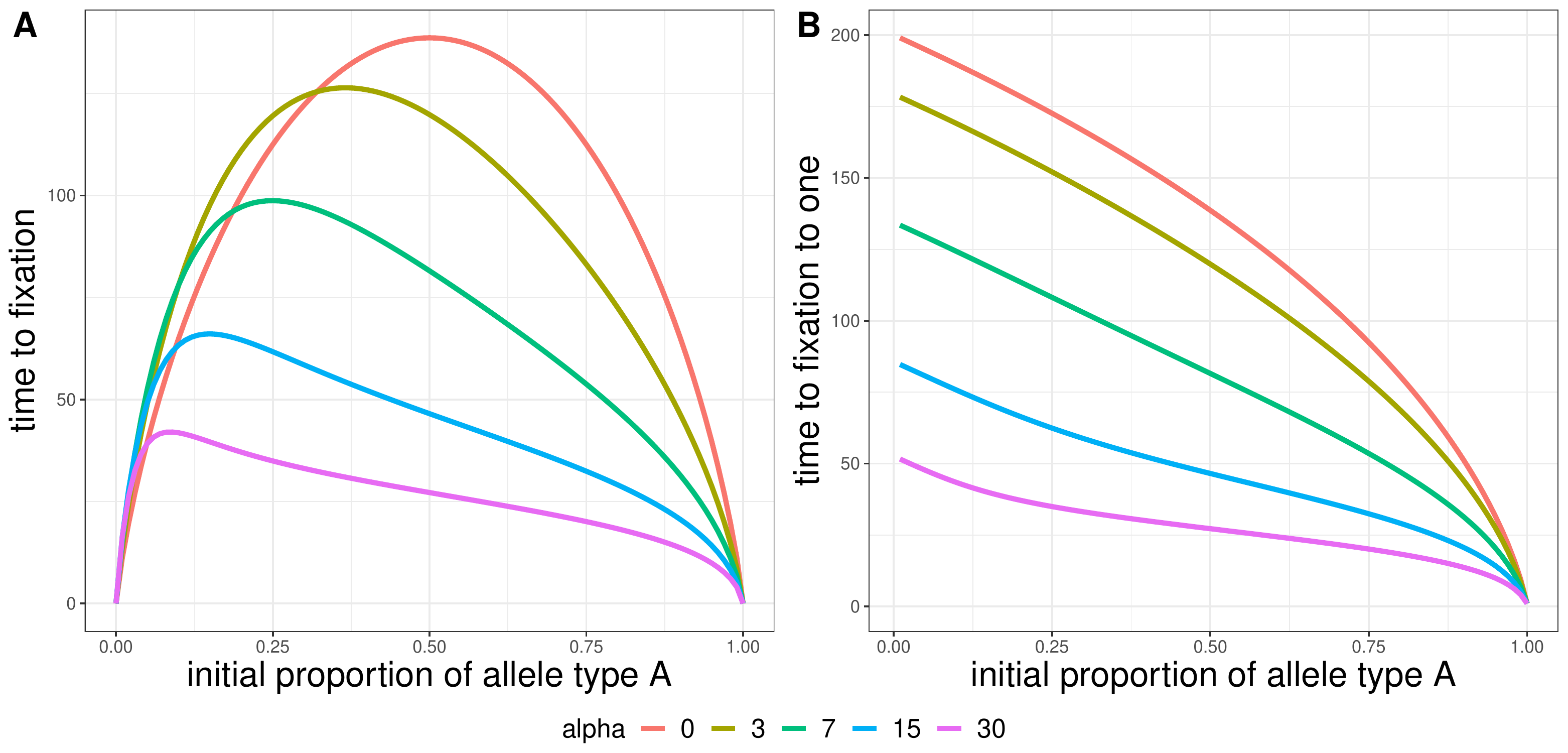} 
\caption{\textbf{(A)} Time-to-fixation function $p \mapsto m(p)$ and \textbf{(B)} fixation-to-one function $p \mapsto m_1(p)$ for various selection effect values and arbitrary population size $N=100$.}
\label{fig:timeToFixation}
\end{figure}

\section{An Approximate Solution}
\label{sec:4}

In our search for simple solutions, we propose an approximate solution for the expression of $m(\cdot)$ involving only elementary functions. 

\begin{theorem}
\label{theoremApprox}
An approximation solution satisfying the boundary conditions is
$$m_{(k)}(p) = 2N \Big(\omega_k(\alpha,p) + \omega_k(-\alpha,1-p)\Big)\,,$$
with
\begin{align*}
\alpha \omega_k(\alpha,p) & =  \Big[\frac{1-e^{- \alpha p}}{1 - e^{- \alpha x}}\Big(-\sum_{n=1}^{k} \frac{(-\alpha x)^n}{n!} \Big(\ln |\alpha x| - H_{n} \Big) \Big)\Big]_{x=p}^{x=1}\,,
\end{align*}
and $H_n$ the n-th harmonic number ($H_n = \sum_{i=1}^{n}\frac{1}{i}$). An upper bound is given by
$$| m(p)- m_{(k)}(p)| \le  2N\frac{e^{\alpha}\alpha^{k}}{k!} \Big(\frac{2|\ln \alpha|+  e^{-1}}{k+1} + 3 \Big)\,.$$
\end{theorem}

\begin{proof}
We simplify the integral term by using the series expansion of the exponential function.
$$    e^{-y} \int_0^{y} \ln |u| e^u du = \int_0^{y} \ln|u| e^{u-y} du = \sum_{n=0}^{+ \infty} \frac{1}{n!} \int_0^{y} \ln|u| (u-y)^n du $$
  $$  = \sum_{n=0}^{+\infty} \frac{1}{n!} \left( \left[\sum_{k=0}^{n} \binom{n}{k} \frac{(-1)^{n-k}}{k+1}\right] y^{n+1} \ln|y| - \left[\sum_{k=0}^{n} \binom{n}{k} \frac{(-1)^{n-k}}{(k+1)^2} \right]y^{n+1} \right) $$
Using relations $\sum_{k=0}^{n} \binom{n}{k} \frac{(-1)^{n-k}}{k+1} = \frac{(-1)^{n}}{n+1}$ and $\sum_{k=0}^{n} \binom{n}{k} \frac{(-1)^{n-k}}{(k+1)^2} = \frac{(-1)^{n}H_{n+1}}{n+1}$ we get:
    $$  e^{-y} \int_0^{y} \ln |u| e^u du  = -\sum_{n=1}^{+\infty} \frac{(-y)^n}{n!} \Big(\ln |y| - H_{n} \Big)\,.
$$
The derivation of the upper bound is exposed in \ref{app:upper}. 
\end{proof}

\begin{remark}
By taking the limit when $\alpha$ tends to zero or by considering $m_{(1)}(p)$ we obtain the classic solution $-2N (p \ln (p) + (1-p) \ln (1-p))$ corresponding to the time to fixation with no selection effect.
\end{remark}

\paragraph{Quality of the Approximate Solution} We compare the relative distance between the exact solution of the fixation time to both zero and one (see Theorem \ref{theorem1}) with our approximate solution. The relative distance is given by the following formula:
$$ D_{k}(p) = \frac{|m(p)-m_{(k)}(p)|}{m(p)}.$$ We denote $k_{0.01}$ the minimum number $k$ of elements in $m_{(k)}(p)$ of our approximate solution such that $D_k(p)<0.01$. Figure \ref{fig:approx} shows $k_{0.01}$ for a series of selection effect values (twenty $\alpha$ values evenly spaced within the interval [1,30] and forty  $\alpha$ values evenly spaced within the interval $[0.1,3]$) and initial probabilities (all values of $p$ within the interval $[0.05, 0.95]$ with a step of size 0.05). We observe that the $k_{0.01}$ increases with the selection effect value $\alpha$. At $\alpha = 0.1$ and marginally to $p$, $k_{0.01} = 1$.

In the literature the selection coefficient $\alpha$ is chosen less than $16$ in the simulation study of Kumira' seminal work \cite{kimura1969average} ($sN \in [0,8]$). It is also common in application studies to consider small selection effects $\alpha < 10$ \cite{Malaspinas2012}. This corresponds to an approximate solution with up to 40 terms valid at precision $0.01$ (see Figure \ref{fig:approx} panel A). It should be noted that considering large $N$ eventually leads to large~$\alpha$. For instance, in the experimental MS2 bacteriophage data \cite{bollback2008estimation} a reasonable set of $N$ is $1\times10^7$ to $2\times10^8$. The corresponding alpha then ranges from $0$ to $1\times10^9$. $k_{0.01}$ is computationally intractable for the upper limit of this range of $\alpha$. However, \cite{bollback2008estimation} reports estimation of this coefficient close to $1$. For $\alpha$ greater than $35-40$ the approximate and the exact solutions are computationally unstable and some new ideas has to be developed. Nevertheless, the need for large selection coefficients $\alpha$ is unclear in the literature as most papers rely on $N$ and $s$ with no understanding of the key role played by the product $sN$.

% \red{How many terms if precision $0.05$ ?}

% common for authors to consider small selection effects $\alpha < 10$ \cite{Malaspinas2012}. This corresponds to an approximate solution with up to 40 terms valid at precision $0.01$ (see Figure \ref{fig:approx} panel A). It should be noted that considering large $N$ eventually leads to a large $\alpha$. For instance, in the experimental MS2 bacteriophage data \cite{bollback2008estimation} a reasonable set of $N$ is $1\times10^7$ to $2\times10^8$. The corresponding alpha ranges from $0$ to $1\times10^9$. $k_{0.01}$ is computationally intractable for the upper limit of this $\alpha$ range. Hence, the approximation loses its interest

%\red{In application studies, most of the authors consider small selection effects $\alpha$ estimated at $0.43$ (see \cite{bollback2008estimation}) in an experimental phage population or in interval $[?,?]$ (see \cite{} in ???). This corresponds to an approximate solution with up to 5 terms valid at precision $0.01$ (see Figure \ref{fig:approx} panel B).} https://doi.org/10.1534/genetics.112.140939 ?

\begin{figure}[H]
    \centering
    \includegraphics[scale=0.5]{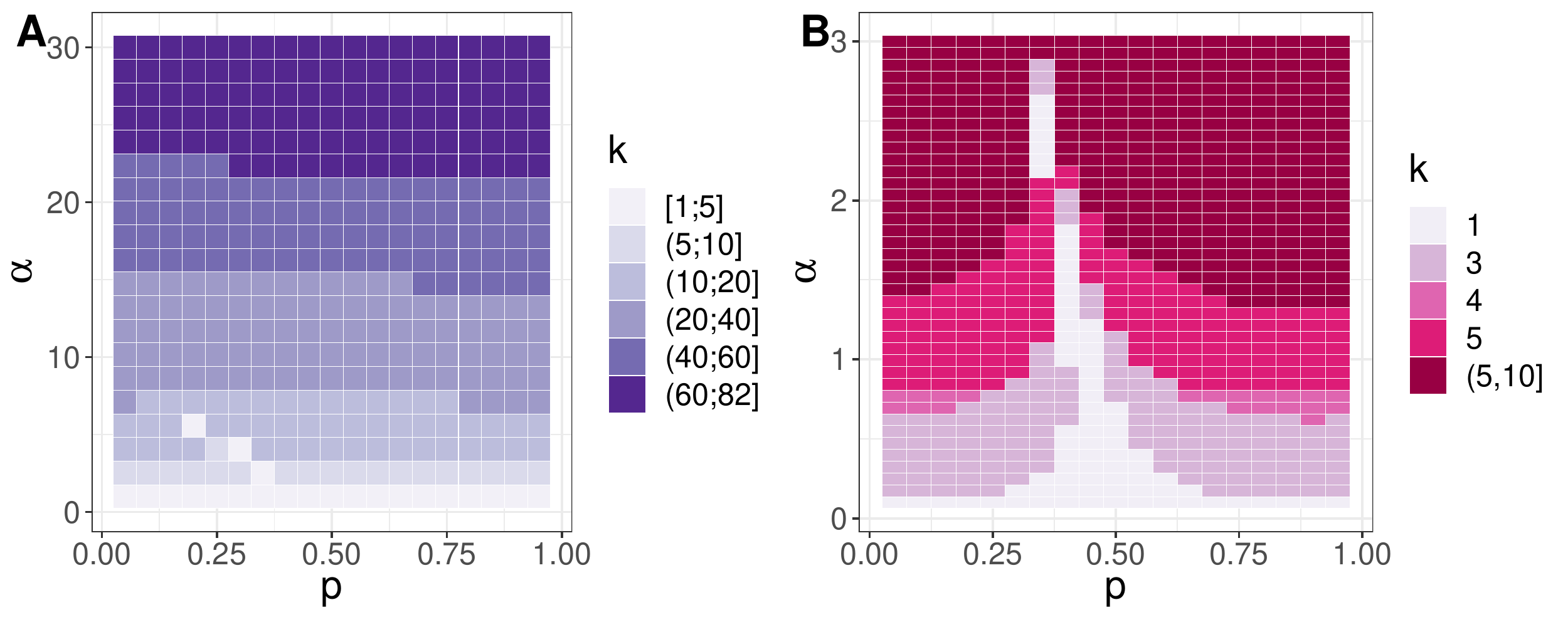}
    \caption{Heatmap of $k_{0.01}$ calculated on a series of selection effects $-$ \textbf{(A)} $\alpha \in [1,30]$ and \textbf{(B)} $\alpha \in [0.1,3]$ $-$ and initial probabilities (see \textit{Quality of the Approximate Solution}). Note that $k_{0.01}=2$ is not observed for these series of $\alpha$.}
    \label{fig:approx}
\end{figure}
%Same as the previous subsection but we have a truncated sum for our approximate result.
\section{Simulation Study}
\label{sec:5}
\subsection{Fixation Time Distribution}
\label{simu1}
To give a better understanding of our problem, we have decided to present the empirical distribution of the fixation time using histograms, separating the two types of fixation (zero and one).
We simulated $10^5$ fixation times with population size $N = 10^4$ for three chosen selection effect values ($\alpha = 1, 5, 10$) and four chosen initial probabilities ($p = 0.1, 0.3, 0.5, 0.7$) in order to appreciate the shape of the fixation time distribution and its variation between fixations to zero and one. These results presented in Figure~\ref{fig:histo} show the impact of the selection effect reducing the probability fixation of allele loss (see function $1-u(\cdot)$). The exact theoretical distribution is not known but some results are known with no selection effect \cite{kimura1970length}. Note that the distribution of fixation time to one seems Log-symmetric independently of the chosen initial probability and selection effect value (see \ref{app:histo_log_scale}).  

\begin{figure}[!h]
\center\includegraphics[scale=0.36]{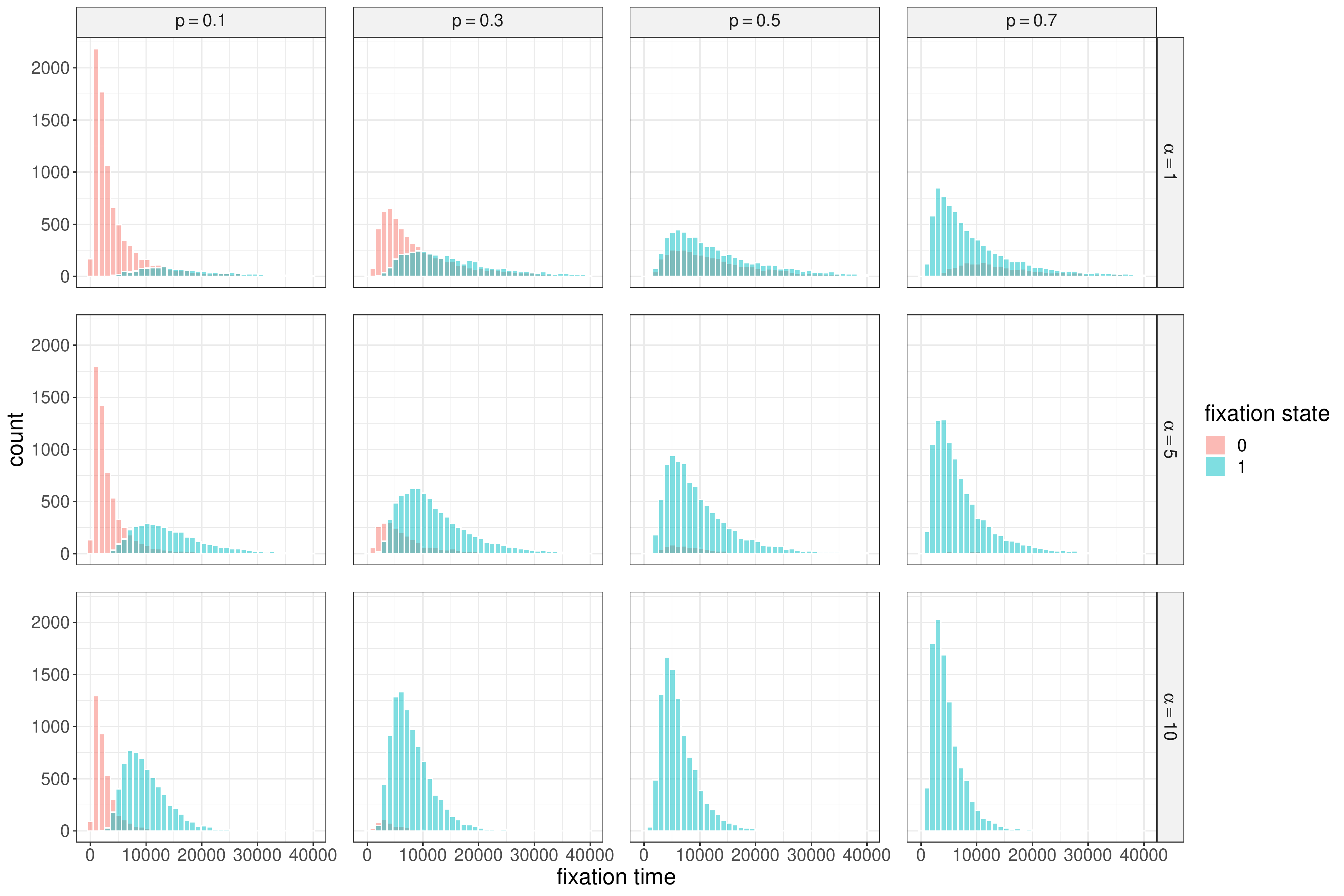} 
\caption{Histograms of relative fixation time to zero (in red) and fixation time to one (in blue). Each histogram represents the result of $10^4$ simulations with population size of $N = 10^4$. The histograms are truncated after a fixation time of 40000.}
\label{fig:histo}
\end{figure}

\subsection{Simulations With Respect to The Selection Coefficient, Population Size and Initial Probability}
\label{simu2}
\paragraph{Design of Simulations} We compare the empirical mean of the fixation time to zero, one and both, obtained by simulations against our exact solution in Theorem \ref{theorem1} (both) and Theorem \ref{theorem2} (zero or one). We test different population sizes $(N = 100, 200, 300, 1000)$, selection effect values (twenty $\alpha$ values evenly spaced within the interval $[1,30]$) and initial probabilities (all values of $p$ within the interval $[0.05,0.95]$ with a step of size $0.05$). For each combination of these parameters we simulated up to $10^5$ repetitions, interrupting the simulations once we get $10^4$ fixations at the desired state (one, zero or both). For some combinations of $\alpha$ and $p$ for fixation to zero, the number of repetitions was not enough to reach $10^4$ fixations.  In that case, we decide to not calculate the empirical mean and set the corresponding values to NA in Figure \ref{fig:relative_distance_sim_vs_analytic_solutions_and_variance_fixation}. This limitation is well understood by the theoretical probability of fixation to zero (function $1-u(\cdot)$).
%For each combination of these parameters we simulate $10^4$ fixations for the desired state (one, zero or both). For some combinations of these parameters we could not calculate the empirical mean of the fixation because we did not have enough observed fixations ($<10^4$). Indeed, it appends when the probability of fixation to zero or one is small and the simulation size not enough large. In our case, we fixed the maximum number of repetition to $10^5$. 
%In Figure \ref{fig:relative_distance_sim_vs_analytic_solutions_and_variance_fixation} the relative distance for these combinations of parameters is set to NA.
\paragraph{Metric} The relative distance between the empirical mean of the fixation time to zero, denoted $\hat m_0(p)$, and the exact solution $m_0(p)$ is given by the following formula 
$$ d_{0}(p) = \frac{|\hat m_0(p)-m_0(p)|}{\hat m_0(p)}.$$ The same way we define $d_{1}(p)$ and $d(p)$, the relative distance involving the fixation time to one and both, respectively. 

\paragraph{Empirical Results} 

In Figure \ref{fig:distrib_relative_distance_sim_vs_analytic_solutions} we observe that for a population size $N=200$ or above, marginally to the selection effects and initial probabilities, the median of $d_1(p)$ is lower than 0.01, i.e 1\% of errors. For $d_0(p)$ and $d(p)$ the same statistical control is reached at population size $N=1000$ and $N=300$, respectively. At population size $N=100$ the distributions are flat which suggests a strong dispersion of relative distances. 
The largest observed relative distance is equal to $0.17$ and is obtained by $d_{0}(p)$ with values $\alpha = 28.47$, $p=0.05$ and $N=100$, leading to absolute values $\hat m_0(p) = 8.70$ and $m_0(p) = 10.18$.
%For instance at $\alpha = 28.47$ and $p=0.05$, $\hat m_0(p)$ and $m_0(p)$ are equals to 8.70 and 10.18 respectively. The corresponding relative distance $d_{0}(p)$ is equal to 0.17 and is the largest observed. 
For this combination of parameters the exact solution we propose is quite far from the simulation results. Taking now $\alpha = 2.53$, $p=0.05$ and $N=100$, $\hat m_0(p)$ and $m_0(p)$ are equal to $138.60$ and $138.48$, respectively. The corresponding relative distance $d_{0}(p)$ is equal to $10^{-3}$ which is a much better estimate. These first results convinced us to explore the relative distance according to the range of tested parameters ($\alpha$ and $p$). In Figure \ref{fig:relative_distance_sim_vs_analytic_solutions_and_variance_fixation} (panel A) we observe that the relative distance is driven by three main effects. (i) \textit{The population size $N$.} Consistent with Figure \ref{fig:distrib_relative_distance_sim_vs_analytic_solutions}, we observe that when $N$ increases, the relative distance decreases. (ii) \textit{The selection effect value $\alpha$.} We observe that when $\alpha$ increases, the relative distance increases. (iii) \textit{The side effect.} We observe that $d_0(p)$ and $d_1(p)$ are larger for small and large $p$, respectively. $d(p)$ is a combination of both $d_0(p)$ and $d_1(p)$. The side effect correlates well with the coefficient of variation (see Figure \ref{fig:relative_distance_sim_vs_analytic_solutions_and_variance_fixation} panel B). Hence, the increases of relative distances in these regions could be explained by an increased uncertainty on the estimate of $\hat m(p)$, $\hat m_0(p)$ and $\hat m_1(p)$.

\begin{figure}[H]
    \centering
    \includegraphics[scale=0.41]{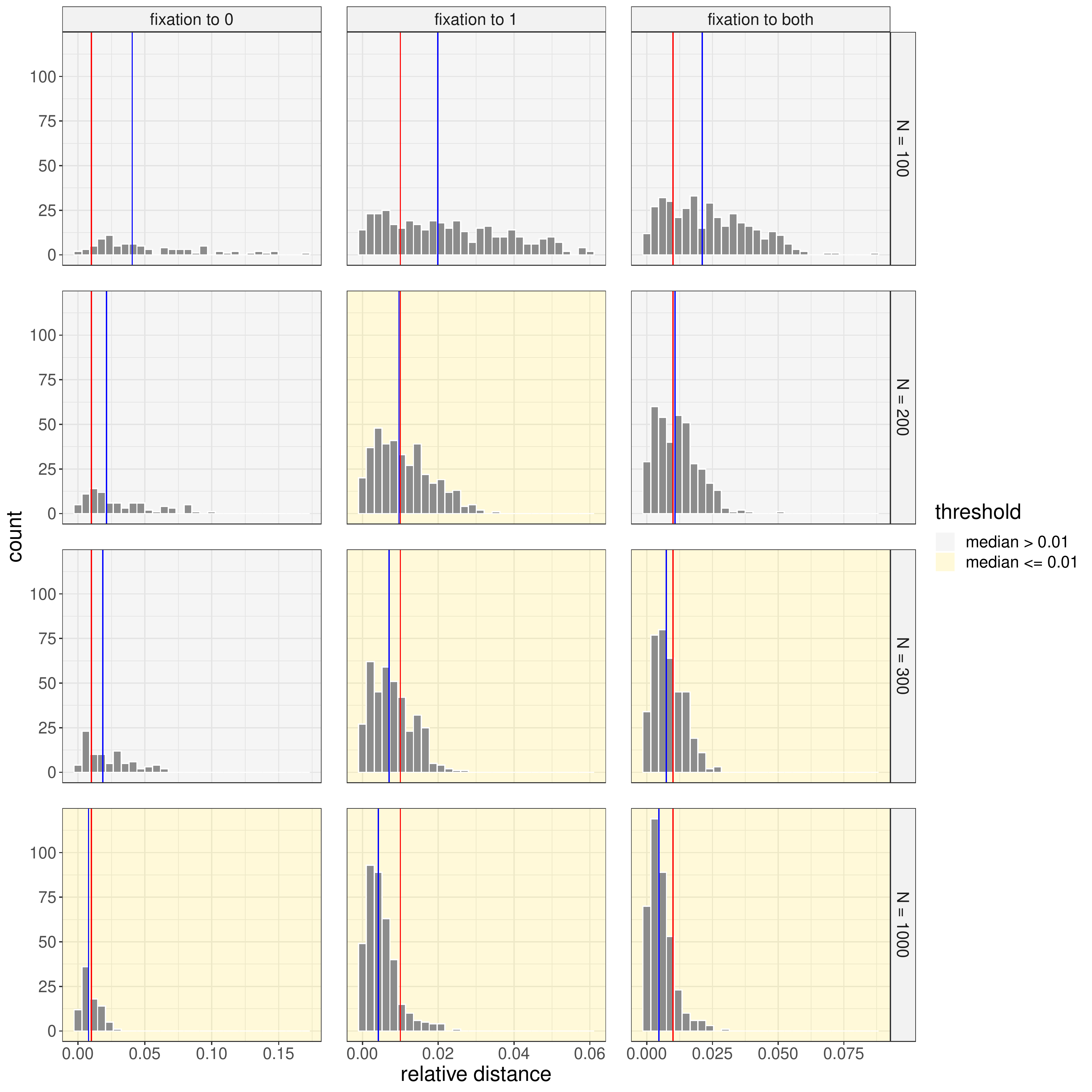}
    \caption{Histogram of relative distances $d_0(p)$, $d_1(p)$ and $d(p)$ calculated on a series of population sizes, selection effects and initial probabilities (see \textit{Design of Simulations}). When the median (blue vertical line) of the relative distance distribution is lower than 0.01 (red vertical line), i.e 1\% of errors, the track background color is set to yellow.}
    \label{fig:distrib_relative_distance_sim_vs_analytic_solutions}
\end{figure}
\begin{figure}[H]
    \centering
    \includegraphics[scale=0.39]{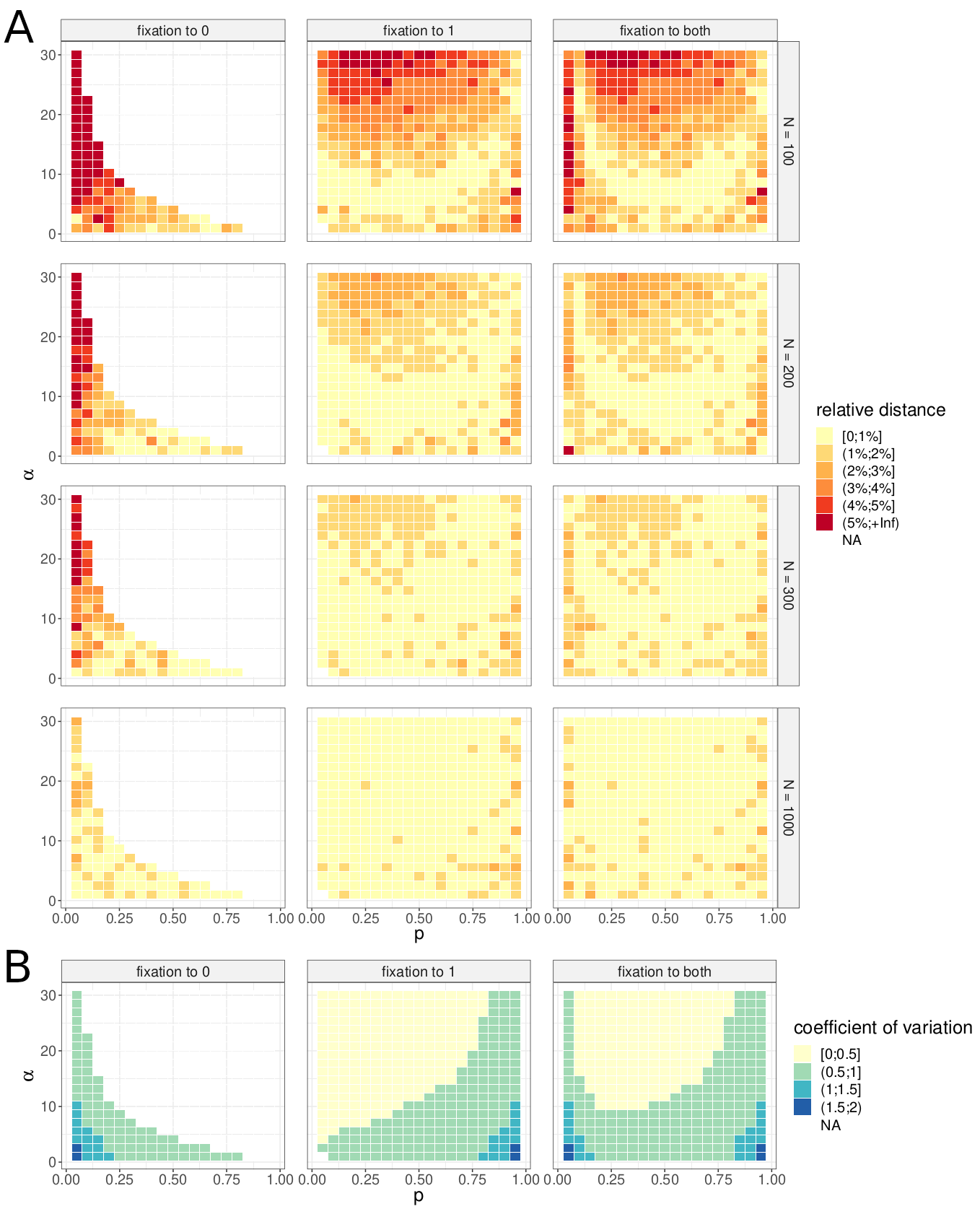}
    \vspace*{-6mm}
    \caption{\textbf{(A)} Heatmap of relative distances $d_0(p)$, $d_1(p)$ and $d(p)$ calculated on a series of population sizes, selection effects and initial probabilities (see \textit{Design of Simulations}). \textbf{(B)} Heatmap of the coefficient of variation estimated on $10^4$ observed fixations for the same series of parameters. The coefficient of variation is independent of the population size $N$ (see \ref{app:variance_fixation_sim}).}
    \label{fig:relative_distance_sim_vs_analytic_solutions_and_variance_fixation}
\end{figure}

\newpage
\section{Conclusion}

Exact and approximate formulas for time to fixation in Wright-Fisher problem with selection have been found. All exact formulas are derived by only using the law of total probability and Taylor expansions. The derivation of these results have then been simplified in comparison to the standard diffusion approximation method based on the Kolmogorov backward partial differential equation. Time to fixation to zero, one and both are only expressed in terms of elementary and exponential integral functions with respect to the initial allele $A$ proportion $p$ in segment $[0,1]$. For small selection effects an approximate solution is possible and highlight its link with the standard solution with no selection. A simulation study have been conducted to highlight the large validity range of our formulas, even for small population size. For future developments building estimators for the selection coefficient based on our exact results would be a possible next step. A better understanding of the higher moments for time to fixation to zero and one as the variance would lead to explicit confidence intervals for this model with selection.

\section*{Contributions and Acknowledgment}

Pauline Spinga has been working three months on this project during her undergraduate internship. She developed an R Shiny app comparing simulation results to the exact fixation time formula \url{https://github.com/PaulineSpinga/wrightfisher}. She also helped finding the approximate solution of Theorem \ref{theoremApprox} and made most of the literature review. Arnaud Liehrmann was responsible for the simulation study and co-developed its related R package \url{https://github.com/vrunge/WrightFisherSelection}. Vincent Runge supervised the project, found the theoretical results and wrote the manuscript. We also would like to thank our colleague Guillem Rigaill (INRAE researcher and Arnaud's PhD supervisor) for helpful comments concerning the simulation study.

%%%%%%%%%%%%%%%%%%%%%%%%%%%%%%%%%%%%%%%%%%%%%%%%%%%%

%\section*{References}

\bibliography{mybibfile}

\appendix

\setcounter{figure}{0} 
\section{Differential Equations for Fixation Time}
\label{app:diffequation}

We describe the proof for the equation in $m(\cdot)$. All other results are proven the same way (we only have to change the constant term). Using a Taylor expansion, we get
\begin{equation}
    \begin{aligned}
    m(p)&= 1 + \int_{\Delta p}  \Prob_{p \rightarrow p+q} m(p + \Delta p)d(\Delta p)\\
    &=  1 + \int \Prob_{p \rightarrow p + \Delta p} \left (m(p) + \frac{m'(p)}{1!}\Delta p + \frac{m''(p)}{2!}\Delta p^2 + O(\Delta p^3)  \right )d(\Delta p)\\
      &\approx  1 + m(p) + \E[\Delta P]m'(p) + \frac{1}{2}\E[(\Delta P)^2]m''(p)\,.\\
    \end{aligned}
\end{equation}
With $\Delta P = (Y_{t+1}-Y_{t})\left.\right|_{Y_{t}=p}$, we have:
$$\E[Y_{t}] = p\,,\quad\E[Y_{t+1}] = \frac{(1+s)p}{1+sp}\,,\quad \E[\Delta P] =  \frac{sp(1-p)}{1+ps}\,,$$
and 
$$\var (\Delta P) = \var (Y_{t+1}) =  \frac{1}{N}\frac{(1+s)p}{1+sp}(1-\frac{(1+s)p}{1+sp}) = \frac{1}{N}\frac{(1+s)p(1-p)}{(1+sp)^2}\,,$$
$$\E[(\Delta P)^2] = \var (\Delta P) + \E[\Delta P]^2 = \frac{p(1-p)}{(1+ps)^2}\left [ \frac{1+s}{N}  + s^2 p(1-p) \right ]\,.$$
We derive approximate formulations for these two first moments by Taylor expansion in power of $\alpha/2N$ using relation $s = \frac{\alpha}{2N}$. We get:
$$
    \E[\Delta P] =
     p(1-p)\frac{\frac{\alpha}{2N}}{1+p\frac{\alpha}{2N}} \simeq p(1-p)\frac{\alpha}{2N}(1-p\frac{\alpha}{2N}) = p(1-p)\frac{\alpha}{2N} + o\Big(\frac{\alpha}{2N} \Big)\,.
$$
Similarly, 
 $$
    \begin{aligned}
   \E[(\Delta P)^2] &= \frac{p(1-p)}{N}  \frac{(1+\frac{\alpha}{2N}) }{(1+p\frac{\alpha}{2N})^2}   + o\Big(\frac{\alpha}{2N}\Big) \,\\
     &=\frac{p(1-p)}{N} (1+\frac{\alpha}{2N}) (1-2p\frac{\alpha}{2N})
     ) + o\Big(\frac{\alpha}{2N}\Big)= \frac{p(1-p)}{N}  + o\Big(\frac{\alpha}{2N}\Big)\,.
    \end{aligned}
$$
Removing the small order terms leads to the desired differential equation.

\newpage

\section{Approximate Solution Upper Bound}
\label{app:upper}

We recall that $\alpha$ is positive.
$$\frac{\alpha}{2N}| m(p)- m_{(k)}(p)| = \Big| \Big[\frac{1-e^{- \alpha p}}{1 - e^{- \alpha x}}\Big(-\sum_{n=k+1}^{+\infty} \frac{(-\alpha x)^n}{n!} \Big(\ln |\alpha x| - H_{n} \Big) \Big)\Big]_{x=p}^{x=1}$$
$$-  \Big[\frac{1-e^{ \alpha (1-p)}}{1 - e^{\alpha x}}\Big(-\sum_{n=k+1}^{+\infty} \frac{(\alpha x)^n}{n!} \Big(\ln |\alpha x| - H_{n} \Big) \Big)\Big]_{x=1-p}^{x=1}\Big| $$
$$\le \sum_{n=k+1}^{+\infty} \frac{\alpha^n}{n!} \Big( |\ln \alpha| + H_{n} \Big) + \sum_{n=k+1}^{+\infty} \frac{(\alpha p )^n}{n!} \Big( |\ln \alpha p| + H_{n} \Big)$$
$$+ (1-p)\sum_{n=k+1}^{+\infty} \frac{\alpha ^n}{n!} \Big( |\ln \alpha | + H_{n} \Big) + \sum_{n=k+1}^{+\infty} \frac{(\alpha (1-p) )^n}{n!} \Big( |\ln \alpha (1-p)| + H_{n} \Big)\,,$$
applying triangle inequalities and using relations $\frac{1-e^{- \alpha p}}{1 - e^{- \alpha}} \le 1$, $\frac{1-e^{ \alpha (1-p)}}{1 - e^{\alpha}} \le 1-p$. With the loose upper bound $H_n \le n$ and $\sum_{n=k+1}^{+\infty} \frac{z^n}{n!} \le e^z\frac{z^{k+1}}{(k+1)!}$ for positive $z$, we get:
$$\frac{1}{2N}| m(p)- m_{(k)}(p)| \le \alpha^{k} \Big[ \frac{e^{\alpha}}{k!} \Big(\frac{|\ln \alpha|}{k+1} + 1 \Big)+$$
$$ \frac{e^{\alpha p} p^{k+1}}{k!} \Big(\frac{|\ln \alpha p|}{k+1} + 1 \Big)+\frac{e^{\alpha(1-p)}(1-p)^{k+1}}{k!} \Big(\frac{|\ln \alpha(1-p)|}{k+1} + 1 \Big)\Big]$$
$$\le  \frac{e^{\alpha}\alpha^{k}}{k!} \Big(\frac{|\ln \alpha|+p^{k+1}|\ln \alpha p|+(1-p)^{k+1}|\ln \alpha (1-p)|}{k+1} + 3 \Big)\,,$$
using the upper bound $e^{\alpha}$ for both $e^{\alpha p}$ and $e^{\alpha (1-p)}$.
Eventually, we have:
$$p^{k+1}|\ln p| \le p p^{k} |\ln p| \le \frac{p}{Me} \le p e^{-1}  \,,$$ and we get
$$| m(p)- m_{(k)}(p)| \le  2N \frac{e^{\alpha}\alpha^{k}}{k!} \Big(\frac{2|\ln \alpha|+  e^{-1}}{k+1} + 3 \Big)\,.$$

\newpage

\section{Fixation time distribution on the $\log$ scale.}
\label{app:histo_log_scale}
\begin{figure}[!h]
\center\includegraphics[scale=0.36]{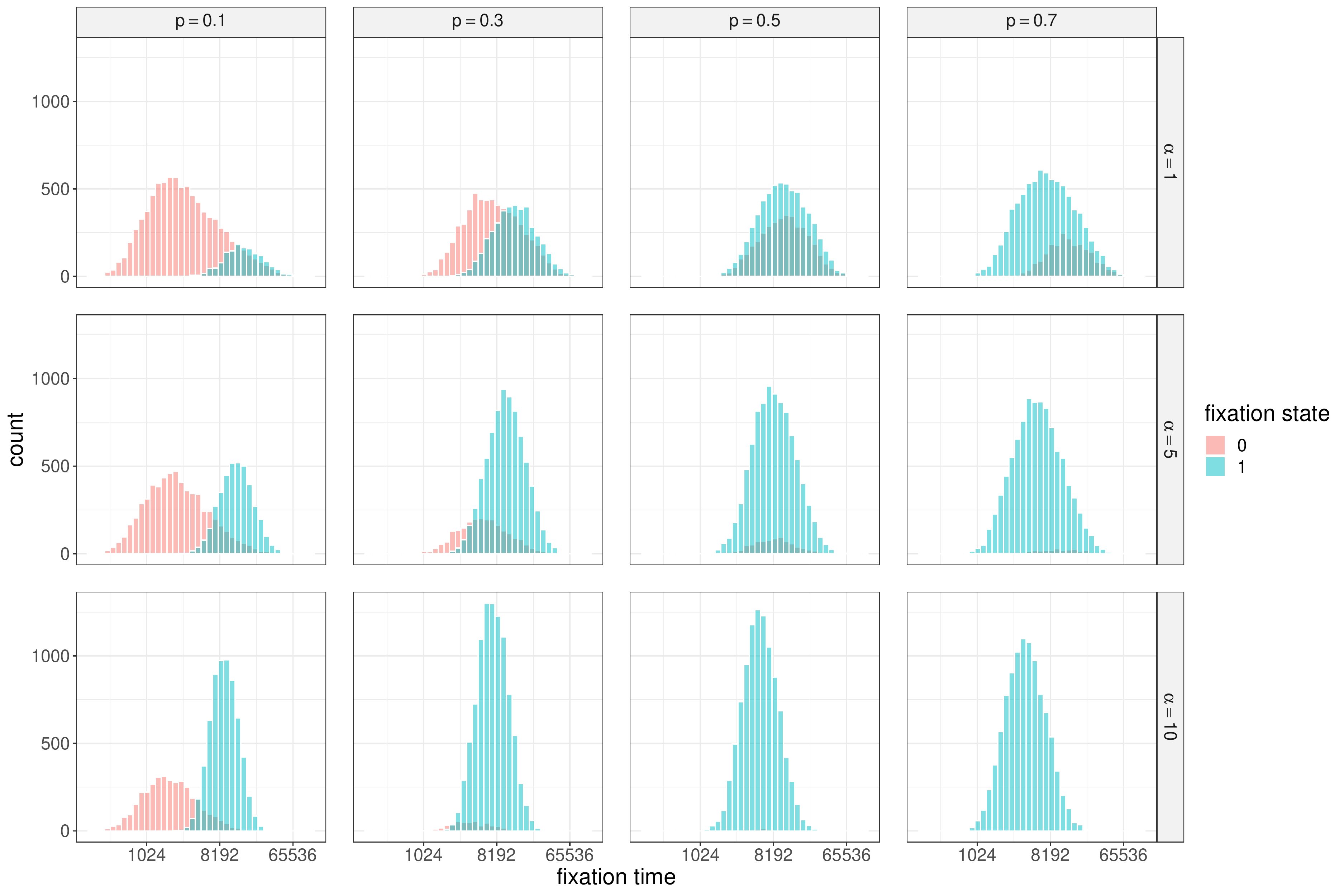} 
\caption{\textit{The distribution of fixation time to one seems Log-symmetric.} Histograms of relative fixation time to zero (in red) and fixation time to one (in blue) on the log scale. Each histogram represents the result of $10^4$ simulations with population size of $N = 10^4$.}
\end{figure}

\newpage

\section{Coefficient of variation is independent of the population size $N$}
\label{app:variance_fixation_sim}
\begin{figure}[!h]
    \centering
    \includegraphics[scale=0.45]{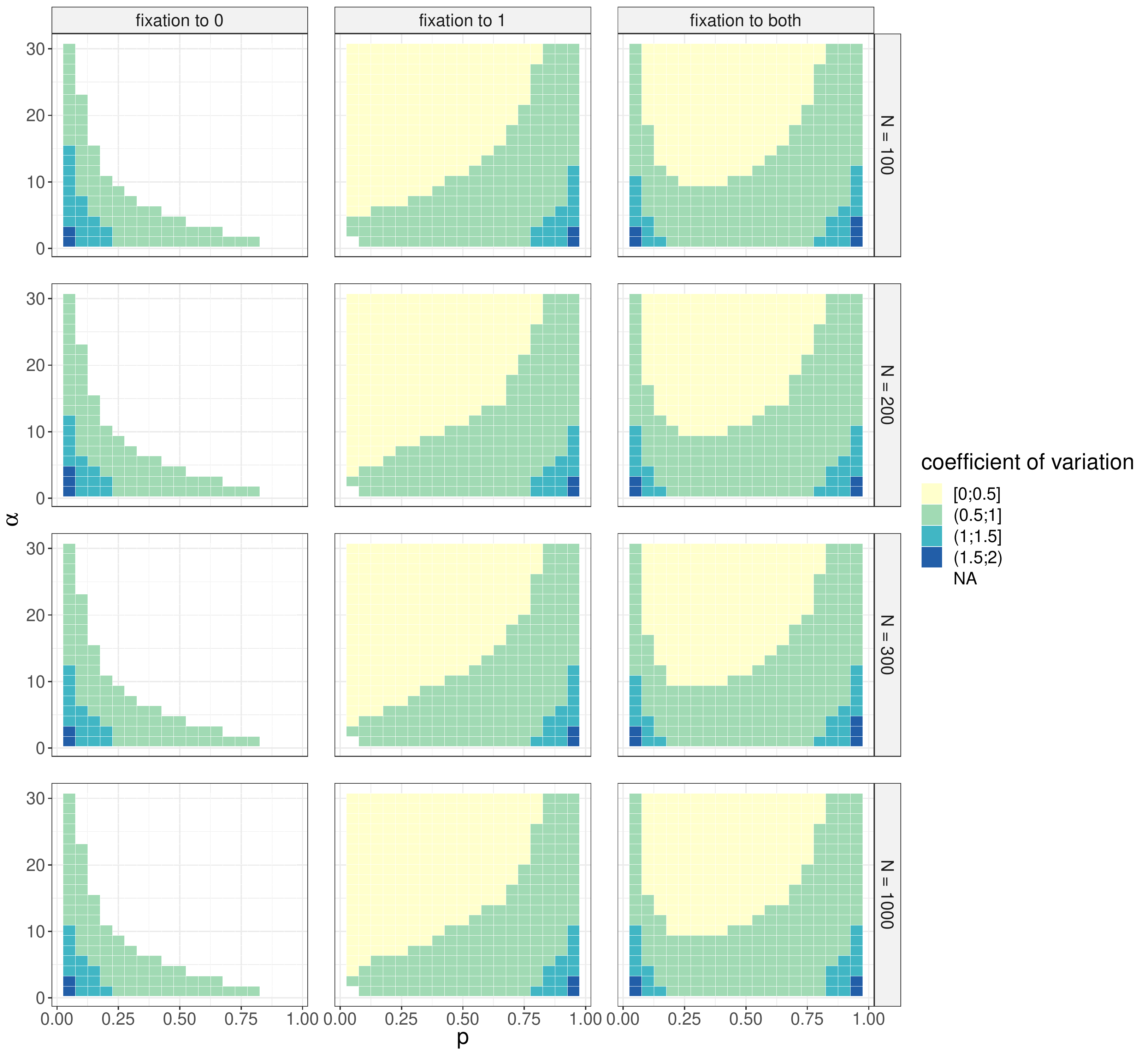}
    \caption{\textit{The coefficient of variation is independent of the population size N.} Heatmap of the coefficient of variation estimated on $10^4$ observed fixations for a series of population sizes, selection effects and initial probabilities.}
    %\label{app:variance_fixation_sim}
\end{figure}

\end{document}